\documentclass{amsart}

\usepackage{amssymb}
\usepackage{amsmath}
\usepackage{amsthm}
\usepackage{amsbsy}
\usepackage{graphicx}
\usepackage{epsfig}
\usepackage{color}
\usepackage{bm}

\theoremstyle{plain}
\newtheorem{theorem}{Theorem}[section]
\newtheorem{lemma}[theorem]{Lemma}

\newtheorem{proposition}[theorem]{Proposition}

\theoremstyle{definition}
\newtheorem{definition}[theorem]{Definition}

\theoremstyle{remark}
\newtheorem{remark}[theorem]{Remark}

\newcommand{\Z}{\mathbb{Z}}
\newcommand{\R}{\mathbb{R}}

\renewcommand{\H}{{\mathbb{H}}}
\newcommand{\F}{\mathcal{F}}

\newcommand{\co}{\colon\thinspace}
\newcommand{\del}{\partial}

\renewcommand{\a}{\alpha}
\renewcommand{\b}{\beta}
\newcommand{\ba}{\bar{\alpha}}
\newcommand{\bb}{\bar{\beta}}

\begin{document} 

\title{Watson-Crick pairing, the Heisenberg group and Milnor invariants}

\author{Siddhartha Gadgil}

\address{	Department of Mathematics\\
		Indian Institute of Science\\
		Bangalore 560003, India}

\email{gadgil@math.iisc.ernet.in}

\thanks{Partially supported by DST (under grant DSTO773) and UGC
(under SAP-DSA Phase IV)}

\date{\today}



\begin{abstract}
We study the secondary structure of RNA determined by Watson-Crick
pairing without pseudo-knots using Milnor invariants of links. We
focus on the first non-trivial invariant,  which we call
the Heisenberg invariant. The Heisenberg invariant, which is an integer,  
can be interpreted in
terms of the Heisenberg group as well as in terms of lattice paths.

We show that the Heisenberg invariant gives a lower bound on the
number of unpaired bases in an RNA secondary structure. We also show
that the Heisenberg invariant can predict \emph{allosteric structures}
for RNA. Namely, if the Heisenberg invariant is large, then there are
widely separated local maxima (i.e., allosteric structures) for the
number of Watson-Crick pairs found.

\end{abstract}

\maketitle

\section{Introduction}

Ribonucleic acid (RNA) is a nucleic acid polymer consisting of
nucleotide monomers, each of which is of one of four types determined
by the nucleotide base present in it. RNA plays a central role in
living cells, specifically in the synthesis of proteins using DNA as a
template. In addition, RNA itself can serve as an information carrier
and also has catalytic properties. Indeed
its ability to serve as both an information carrier and a catalyst has
led to speculation that life began as an \emph{RNA world}~\cite{RNA}.

The \emph{primary structure} of an RNA molecule is the sequence of
nucleotide bases in it. In addition to this, the properties of an RNA
molecule depend strongly on the \emph{secondary structure}, which
is the $3$-dimensional shape of the molecule~\cite{FW}.

The bases form two complementary pairs, with members of a pair forming
strong hydrogen bonds. Thus, if two subsequences of the RNA sequence
are complementary, a stable secondary structure called a \emph{stem
loop structure} can be formed by bonds between complementary pairs in
these subsequences. We study here these \emph{secondary structures} of
an RNA molecule, determined by the Watson-Crick pairing. We shall
consider structures without so called \emph{pseudo-knots}. For basic
concepts for RNA folding, we refer to~\cite{Ti} and~\cite{RNA}. For
surveys and other studies of RNA folding, we refer
to~\cite{RNA}--\cite{c14}. Henceforth by secondary structure we mean
the secondary structure determined by Watson-Crick pairing without
pseudo-knots.

We focus here on introducing new methods to yield a conceptual
understanding of RNA folding. The model we consider - Watson-Crick
pairing without pseudo-knots, is clearly an approximation in various
ways. Firstly, stereo-chemical forces do not allow very short
loops. Secondly, pseudo-knots are present in nature. Thus, allowing
short loops but not pseudo-knots is an approximation of stereo-chemical
forces. Further, we do not take into account the difference between
the strengths of the A-U and the G-C bonds, and also ignore the
non-Watson-Crick bonds.

Thus, our model is clearly not appropriate for the computational study
of individual RNA molecules. Our goal is rather to introduce new
methods for understanding RNA secondary structure, and show that these
are very powerful in the context of our model. One can extend these
methods to take into account more realistic models of RNA secondary
structure.

Our methods can be motivated by a simple observation -- if, for
example, there are more Adenine than Uracil bases, then some Adenine
bases must be unpaired. These considerations give a very elementary
lower bound on the number of unpaired bases. Mathematically, this can
be viewed as coming from abelianisation. Another context in which
abelianisation gives the simplest criteria is the \emph{linking
number} for classical links. In this paper, we show that one can adapt
Milnor's theory of \emph{higher linking numbers} to the context of
RNA.

We use a natural model for Watson-Crick pairing of RNA (without
pseudo-knots) in terms of the free group $F$ on two generators
$\alpha$ and $\beta$. An element of the free group is given by a word
in the four letters $\alpha$, $\bar{\alpha}$, $\beta$ and
$\bar{\beta}$, with $\ba$ and $\bb$ the inverses of $\a$ and $\b$
respectively. We identify these letters with the nucleotides Adenine,
Uracil, Guanine and Cytosine respectively. Under this identification,
an RNA molecule gives a string in the four letters $\alpha$, $\beta$,
$\ba$ and $\bb$.

Stem loops, which are the basic units of RNA secondary structure, then
correspond to words in the free group of the form $gl\bar{g}$, with
$g$ and $l$ words in the free group and $\bar{g}$ the inverse of $g$
(see figure~\ref{stemloop}). One may further have Watson-Crick pairing
within the word $l$, so that a subword of the RNA sequence may be of
the form $ga(hl\bar{h})b\bar{g}$ as in figure~\ref{second}. In
general, a secondary structure without pseudoknots consists of
pairings between letters and their inverses so that there is no
\emph{nesting}. We formalise such a structure (which we call a
\emph{folding}) in Definition~\ref{deffold}. The appropriate
\emph{energy}, whose local and global minima we study, is the number
of unpaired bases in the secondary structure. 

\begin{figure}
\input{stem.pstex_t}
\caption{A stem loop with bases labelled using both nucleotides and
letters in the free group}\label{stemloop}
\input{second.pstex_t}
\caption{An RNA secondary structure}\label{second}
\end{figure}

Mathematically, one can interpret the above as saying that the number
of unpaired bases is a \emph{conjugacy invariant norm}. For, a
secondary structure on RNA corresponding to the word $w$ gives one for
the word $gw\bar{g}$ (with the initial segment corresponding to $g$
paired with the final segment corresponding to $\bar{g}$) so that the
number of unpaired bases is the same. Further, given secondary
structures on strands of RNA corresponding to words $w_1$ and $w_2$,
we get a secondary structure on the strand corresponding to $w_1w_2$
(the concatentation of the words). This paper is based on the
observation that conjugacy invariance can be taken into account by
considering nilpotent quotients.

In analogy with Milnor invariants, we shall associate numbers to words
$g$ in the letters $\a$, $\b$, $\ba$ and $\bb$, which we call
\emph{invariants} according to terminology familiar in mathematics
(they are not to be taken as invariant in any biological sense). By
the above, it is desirable that they are invariant under conjugation,
or at least the extent to which they are not invariant can be
estimated. To achieve sub-additivity, we consider additive functions
of $g$ and take their absolute value, or more generally sums of the
absolute values of such functions.

As we have seen, an obvious lower bound for the number of unpaired
elements is given by comparing the number of letters that are $\alpha$
with the number that are $\ba$ and similarly for $\beta$ and
$\bar{\beta}$. We denote the difference between the number of letters
of a word $g$ that are $\a$ and the number that are $\ba$ by
$a(g)$. Similarly, $b(g)$ denotes the difference between the number of
letters that are $\b$ and the number that are $\bb$. In other words we
look at the image of $g$ under the abelianisation map
$ab:F\to\Z^2$. Then $a(g)$ and $b(g)$ denote the co-ordinates of
$ab(g)$ so that $ab(g)=a(g)ab(\alpha)+b(g)ab(\beta)$. The numbers
$a(g)$ and $b(g)$ are the first of our invariants.

Clearly, the minimum number of unpaired bases in $g$ is at least
$\vert a(g)\vert+\vert b(g)\vert$, but this is very far from sharp. We
proceed further in analogy with Milnor's theory of link homotopy. The
simplest invariants of a link are the linking numbers. These are given
by considering the image of a curve in the abelianisation of an
appropriate fundamental group, and are thus analogous to $a(g)$ and
$b(g)$. Milnor constructed \emph{higher linking numbers} by
considering appropriate nilpotent quotients.

We shall associate an appropriate link to RNA molecules in
Section~\ref{S:mil} and construct various invariants. The main focus
in this paper, however, is to construct and study the first of these
higher invariants. It is easiest to proceed with a direct algebraic
description. This description is in terms of another familiar object
-- the Heisenberg group. We thus call this invariant the
\emph{Heisenberg invariant} $\nu(g)$.  There is also a nice geometric
view of this invariant in terms of areas enclosed by \emph{lattice
paths}. We also provide an elementary combinatorial description in
Theorem~\ref{comp} which allows for easy computation.

We show that the Heisenberg invariant gives a lower bound on the
number of unpaired bases in an RNA molecule, and hence the potential
energy of a secondary structure.

We also show that the Heisenberg invariant is related to a
biologically significant property of Watson-Crick pairing. Namely, we
show that if $\nu(g)$ is \emph{sufficiently large}, then there are
\emph{allosteric structures}, i.e., local minima for the number of
unpaired bases that are \emph{widely separated}. This means that there
are two ways of \emph{folding} the sequence so that we cannot pass
from one to the other without significantly increasing the number of
unpaired bases (here a folding is an abstraction of the secondary
structure). Thus, the Heisenberg invariant as well as the higher
invariants should prove very fruitful in the study of RNA secondary
structures.

Our results depend on our simplified model. However, it is easy to see
that similar results continue to hold even if one makes the model
biologically more realistic by taking into account that nearby bases
do not pair (i.e., there are no very short loops). In the case of the
lower bound this is obvious, as further restrictions can only increase
the number of unpaired bases. The result concerning allosteric
structures also extends as we sketch following the proof of the result
for our model.

\section{The Heisenberg invariant}

The Milnor higher link invariants~\cite{Mi1}\cite{Mi2} are based on
the lower central series of a group. We recall some basic definitions.

Consider a group $G$. For elements $a,b\in G$, $\bar{a}$ denotes the
inverse of $a$ and $[a,b]$ denotes the commutator
$ab\bar{a}\bar{b}$. For subgroups $H_1,H_2\subset G$, we define
$[H_1,H_2]$ to be the normal subgroup generated by elements of the
form $[a,b]$, $a\in H_1$, $b\in H_2$. 

The lower central series of $G$ is defined inductively as follows. Let
$G_1=G$. If $G_n$ has been defined, we define $G_{n+1}$ to be
$[G,G_n]$. In particular, $G_2=[G,G]$ and the abelianisation of $G$ is
$G/G_2$. Note that $G_1\supset G_2\supset G_3\supset\dots$.

Consider now the free group $F$ generated by $\a$ and $\b$. Let $H$ be
the group $F/F_3$. The Heisenberg invariant is obtained by considering
the image $[g]\in H$ of an element $g\in F$. It is well known that
this is the Heisenberg group, which is the unique central extension of
$\Z^2$ by $\Z$. For the convenience of the reader, we prove these
properties below.

\begin{proposition}\label{exact}
There is an exact sequence
\begin{equation}\label{eqexact}
1\to \Z\to H\to \Z^2\to 1
\end{equation}
with the image of $\Z$ central in $H$ and generated by
$[\a,\b]$. Further, $H$ is isomorphic to the Heisenberg group.
\end{proposition}
\begin{proof}
As $F_3\subset F_2$, the abelianisation homomorphism $ab:F\to
\Z^2=F/F_2$ induces a surjective homomorphism $\varphi:H=F/F_3\to
F/F_2=\Z^2$. It is well known that the kernel of $ab:F\to \Z^2$ is the
normal subgroup in $F$ generated by $[\a,\b]$. Hence the kernel of
$\varphi$ is the normal subgroup in $H$ generated by the equivalence
class of $[\a,\b]$ in $H$, which we continue to denote by $[\a,\b]$.

Note that the commutators $[[\a,\b],\a]$ and $[[\a,\b],\b]$ are both
elements of $F_3$, and hence have trivial images in $H$. It follows
that $[\a,\b]\in H$ commutes with the images of $\a$ and $\b$ in $H$
and hence is central. Thus, as the kernel of $\varphi$ is the normal
subgroup generated by $\varphi$, it is in fact the cyclic group
generated by $[\a,\b]$. To complete the proof of the exact sequence of
Equation~\ref{eqexact}, it suffices to show that no power of
$[\a,\b]$ is trivial in the group $H$.

We show this by constructing an explicit homomorphism from $H$ to the
Heisenberg group $\mathbb{H}$ with integer coefficients, namely the
group of matrices of the form
\begin{equation}
M(a,b,c)=\left(\begin{array}{ccc} 1 & a & c\\ 0 & 1 & b\\ 0 & 0 &
1\end{array}\right)
\end{equation}
with $a$, $b$ and $c$ integers.

The homomorphism $\psi:H\to \mathbb{H}$ is defined as follows. Let
$\Psi\co F\to\mathbb{H}$ be the unique homomorphism taking $\a$ and
$\b$ to the matrices $M(1,0,0)$ and $M(0,1,0)$. By a well known (and
straightforward) computations, $\Psi([\a,\b])= M(0,0,1)$, and
$M(0,0,1)$ is central in $\mathbb{H}$. Hence $\Psi(F_3)$ is the
trivial group.  Thus, we get a well-defined homomorphism $\psi\co
H=F/F_3\to \mathbb{H}$, which is clearly surjective. As
$\psi([\a,\b])=M(0,0,1)$, it follows that
$\psi([\a,\b]^k)=M(0,0,k)$. Hence if $k\neq 0$, $[\a,\b]^k$ is
non-trivial as an element of $H$. This shows that the sequence of
Equation~\ref{eqexact} is exact.

Finally, we show that $\psi$ is injective, hence an
isomorphism. Suppose $\psi(g)=1$. Then as $\psi$ gives an isomorphism
on the abelianisations of $H$ and $\H$, it follows that
$ab(g)=0$. Hence by the exact sequence, $g=[a,b]^k$ for some
$k\in\Z$. It follows that $\psi(g)=M(0,0,k)$, hence $\psi(g)=1\implies
k=0\implies g=1$.
\end{proof}

The homomorphisms $a\co F\to \Z$ and $b\co F\to \Z$ defined on $F$
factor through $H$ (as $F_3\subset F_2$), and we continue to denote
them by $a(\cdot)$ and $b(\cdot)$. We shall also denote the images of
$\a$ and $\b$ in $H$ by $\a$ and $\b$.

It is easy to see that for an element $g\in H$, $a(g)$ and $b(g)$ are
the entries $a$ and $b$ of $\psi(g)$. We can define the Heisenberg
invariant in terms of the remaining entry $c$. However, it will be
convenient to take a different approach based on the following
proposition (which is a special case of normal forms that are well
known in the literature).

\begin{proposition}\label{rep}
Any element $g\in H$ can be uniquely expressed as 
$$g=[\a,\b]^\nu \a^a \b^b$$
with $\nu$, $a$ and $b$ integers.
\end{proposition}
\begin{proof}
Let $h=g\b^{-b(g)}\a^{-a(g)}$. Then the image under the abelianisation
map $\varphi(h)$ of $h$ is trivial. Hence $h$ is in the kernel of
$\varphi$. By Proposition~\ref{exact}, $h=[a,b]^\nu$ for a unique
$\nu$. Hence it follows that $g=[\a,\b]^\nu \a^a \b^b$ with $a=a(g)$
and $b=b(g)$.

To see uniqueness, observe by abelianising that if $g=[\a,\b]^\nu \a^a
\b^b$, we must have $a=a(g)$ and $b=b(g)$. Further, if
$h=g\b^{-b}\a^{-a}$, then $\nu$ is the unique integer such that
$h=[a,b]^\nu$, and hence is determined by $g$.
\end{proof}

\begin{definition}
The \emph{Heisenberg invariant} of $g\in H$ is the unique $\nu$ such
that $g=[\a,\b]^\nu \a^a \b^b$ with $a$ and $b$ integers. The
Heisenberg invariant of a word in $F$ is the Heisenberg invariant of
the image of the word in $H$. 
\end{definition}

In terms of the above definition, the representation of
Proposition~\ref{rep} can be expressed as
\begin{equation}\label{eqrep}
g=\a^{a(g)}\b^{b(g)}[\a,\b]^{\nu(g)}
\end{equation}
The higher link invariants are defined only for links with trivial
linking number. In our situation, we have defined the Heisenberg
invariant for all words. Nevertheless, it should be regarded as
well-defined up to an error given by the abelianisation
$(a(g),b(g))$. For instance, in Proposition~\ref{rep}, we can
interchange the order of $\beta$ and $\alpha$ in the representation of
$g$. An analogous result still holds but we get a different value of
the Heisenberg invariant.

We have the following simple properties of the Heisenberg
invariant. In what follows all equalities are to be understood to be
in the group $H$.

\begin{proposition}
If $g_1,g_2\in H$ have trivial abelianisations,
$\nu(g_1g_2)=\nu(g_1)\nu(g_2)$.
\end{proposition}
\begin{proof}
As $g_1$ and $g_2$ have trivial abelianisation, by
Proposition~\ref{exact} $g_i=[\a,\b]^{k_i}$, $i=1,2$, for some
integers $k_i$. By definition, $\nu(g_i)=k_i$. Further,
$g_1g_2=[\a,\b]^{k_1+k_2}$, hence
$\nu(g_1g_2)=k_1+k_2=\nu(g_1)+\nu(g_2)$.
\end{proof}

\begin{proposition}
If $g\in H$ has trivial abelianisation, then for $h\in H$,
$\nu(hg\bar h)=\nu(g)$.
\end{proposition}
\begin{proof}
By Proposition~\ref{exact}, as $g$ is in the kernel of the
abelianisation homomorphism $\varphi$, $g$ is central. Thus,
$hg\bar{h}=g$
\end{proof}

For the sake of clarity, we shall focus on elements  $g\in H$ with
trivial abelianisation. However, it is easy to obtain variants of all
our results allowing for errors determined by $a(g)$ and $b(g)$ using
the following proposition.

\begin{proposition}\label{unbal}
For $k\in\Z$, the following identities hold.
\begin{enumerate}
\item $\nu(\a^kg)=\nu(g)$
\item $\nu(g\b^k)=\nu(g)$
\end{enumerate}
\end{proposition}
\begin{proof}
As $g=[\a,\b]^{\nu(g)}\a^{a(g)}\b^{b(g)}$ and $[\a,\b]$ is central, 
$\a^kg=[\a,\b]^{\nu(g)}\a^{k+a(g)}\b^{b(g)}$, hence
$\nu(\a^kg)=\nu(g)$. Similarly, 
$g\b^k=[\a,\b]^{\nu(g)}\a^{a(g)}\b^{b(g)+k}$, hence $\nu(gb^k)=\nu(g)$.
\end{proof}

We shall say that the word $g$ is \emph{balanced} if $a(g)=b(g)=0$.

\section{Lattice paths and Area}

In this section, we give a geometric interpretation of the Heisenberg
invariant, and some obvious extensions that follow from this. Examples
viewed in this fashion form the intuition for the rest of the paper. As
the results here are not used formally elsewhere, we shall skip most
proofs.

The plane $\mathbb{R}^2$ contains the lattice $\mathbb{Z}^2$. We can
associate to each word $g=l_1 l_2\dots l_n$ in the letters $\alpha$,
$\ba$, $\b$ and $\bb$ a path in the plane as follows. We start at the
origin $(0,0)$. In the first step, we take a path from $(0,0)$ to one
of the points $(1,0)$, $(0,1)$, $(-1,0)$ and $(0,-1)$ according as
$l_1$ is $\a$, $\b$, $\ba$ or $\bb$. Inductively, at the end of the
$(k-1)$th step we will have a path from $(0,0)$ to a lattice point
$(p,q)\in\Z^2$. We extend the path by a unit segment joining $(p,q)$
to one of the points $(p,q)+(1,0)$, $(p,q)+(0,1)$, $(p,q)+(-1,0)$ and
$(p,q)+(0,-1)$ according as $l_k$ is $\a$, $\b$, $\ba$ or $\bb$.

Thus, we obtain a path in the plane consisting of horizontal and
vertical segments. If $a(g)=b(g)=0$, this path is a loop $\gamma$. The
Heisenberg invariant is the \emph{oriented area}, interpreted
appropriately, bounded by the loops $\gamma$. If $\gamma$ is a simple
loop, it bounds a region $R$. In this case, $\nu(g)$ is $\pm Area(R)$,
with the sign determined by whether $\gamma$ is a counterclockwise or
a clockwise loop around $R$. In general, we can regard $\gamma$ as the
boundary of a region $R$, allowing signs and multiplicities, which we
define below in terms of winding numbers. The area of $R$, taking into
account the signs and multiplicities, is the Heisenberg invariant
$\nu(g)$.

Consider a unit square $\Delta(p,q)$, $p,q\in\Z^2$, with vertices
$(p,q)$, $(p+1,q)$, $(p,q+1)$ and $(p+1,q+1)$. Let $c(p,q)$ be the
winding number of $\gamma$ about the centre $z$ of $\Delta(p,q)$. We
remark that we can take the winding number about any interior point to
get the same result. The integer $c(p,q)$ is the (possibly negative)
multiplicity of the square $\Delta(p,q)$.

Note that only finitely many of these coefficients are non-zero. We
can interpret the Heisenberg invariant as (the finite sum)
$$\nu(g)=\sum_{(p,q)\in\Z^2} c(p,q)$$

In these terms, there are obvious extensions of the Heisenberg
invariant. The group $\Z^2$ acts on the plane by translations. We use
multiplicative notation for $\Z^2$ and denote generators by $s$ and
$t$, so that the action of $s$ is translation by $(1,0)$ and that of
$t$ is translation by $(0,1)$. Let $\Delta=\Delta(0,0)$. Then
$\Delta(p,q)$ is the image $s^pt^q\Delta$ of $\Delta$. The region $R$
can be expressed as the formal sum
$$\sum_{(p,q)\in\Z^2} c(p,q)s^pt^q\Delta$$

Thus, we associate to $g$ the polynomial in two variables
$$P_g(s,t)=\sum_{(p,q)\in\Z^2} c(p,q)s^pt^q$$

The Heisenberg invariant is $P_g(1,1)$. Extensions are given by other
(linear) functions of the polynomial $P_g$. As mentioned in the
introduction, it is natural to consider such functions $F(P_g)$ so
that if $g$ is conjugate to $g'$, then $F(P_g)=F(P_{g'})$.  This
translates to being invariant under multiplication by the
polynomial $s^kt^l$ for $k,l\in\Z$, i.e., for a polynomial $P$,
$F(P)=F(s^kt^lP)$. It is clear that the Heisenberg invariant $P_g(1,1)$
has this property. We next show how to construct secondary invariants,
i.e., which are invariant under conjugation provided the Heisenberg
invariant vanishes.

Let $p_g(s)=P_g(s,1)$. Then $p_g(s)$ is invariant under multiplication by
any power of $t$. Let $\nu_s(g)=p_g'(1)$. Suppose $\nu(g)$
vanishes. Then for an element $g'$ conjugate to $g$, $P_{g'}(s,t)$ is
of the form $s^kt^lP_g(s,t)$. Hence,

$$p_{g'}'(1)=(s^kp_g)'(1)=(ks^{k-1}p_g)(1)+ (s^kp_g')(1)=p_g'(1)$$
as $p_g(1)=\nu(g)=0$ by hypothesis.

We have a similar invariant taking $t$ in place of $s$. In case these
invariants vanish, we get further invariants by taking higher
derivatives.

We next turn to the general case, where we do not necessarily have
$a(g)=b(g)=0$. As before, we get a path from $(0,0)$ to
$(a(g),b(g))$. We make this into a loop $\gamma$ by extending this by
the vertical segment to $(a(g),0)$ and then the horizontal segment to
the origin. The Heisenberg invariant is then the area enclosed by this
loop. Note that there are other minimal paths joining $(a(g),b(g))$ to
the origin, which give different values for the Heisenberg
invariant. Hence $\nu(g)$ should be regarded as defined up to
indeterminacy given by $a(g)$ and $b(g)$.

We give another interpretation of the above in homological terms. This
will not be used in the sequel.

The plane $\R^2$ has a natural structure as a cell complex $X$ with
vertices lattice points, edges horizontal or vertical unit segments
joining adjacent lattice points and faces unit squares. The group
$\Z^2$ acts on this cell complex freely by translations. The quotient
$Y$ of the one-skeleton $X^{(1)}$ under this action is the wedge of
two circles. This has fundamental group $F$, and the one-skeleton of
$X$ is the (Galois) cover corresponding to the subgroup $F_2=[F,F]$.

Any word $g$ in the free group gives a path in $Y$. This lifts to a
path in $\gamma$ starting at the origin which can be regarded as a
$1$-chain in $C_1(X)$. If $a(g)=b(g)=0$, then $\gamma$ is a loop as
$g\in F_2$, and hence is a $1$-cycle. As the plane is contractible,
this is a boundary $\gamma=\del\zeta$, $\zeta\in C_2(X)$. As $H_2(X)$
and $C_3(X)$ are trivial, it follows that $\zeta$ is unique.

Let $\Delta$ be a fixed unit square. Any other unit square is the
image $g\Delta$ of $\Delta$ under the action of $\Z^2$ on $C_2(X)$,
and $g$ is unique. As the unit squares are a basis of $C_2(X)$, we can
uniquely express $\zeta$ as a finite sum. 
$$\zeta=\sum_i n_ig_i\Delta,\ n_i\in\Z,g_i\in\Z^2$$

We can interpret the Heisenberg invariant as
$$\nu(g)=\sum_i n_i $$

This has obvious extensions. We note that we can associate to $g$ the
element $\sum_i n_ig_i$ in the group ring $\Z[\Z^2]$ (which
corresponds to the polynomial $P_g(s,t)$). This is well-defined (but
is natural only up to multiplication by an element of the group
$\Z^2$). We have considered the image of this element under the
homomorphism $\Z[\Z^2]\to \Z$ taking each element of the group $\Z^2$
to $1$. We can obviously obtain more refined estimate by considering
either the full group ring, or at least other representation of the
group ring.

\section{Identities for the Heisenberg invariant}

We collect in this section some elementary identities in the group $H$
and formulae for the Heisenberg invariant. Recall that the element
$[\a,\b]$ is central in the group $H$.

\begin{lemma}\label{ids}
The following identities hold in $H$.
\begin{enumerate}
\item $\a\b\ba=[\a,\b]\b$.
\item For $k\in \Z$, $\a\b^k\ba=[\a,\b]^k\b^k$.
\item $\b\a\bb=[\a,\b]^{-1}\a$.
\item For $k\in \Z$ $\b\a^k\bb=[\a,\b]^{-k}\a^k$.
\item For $g\in H$, $\a g\ba=[\a,\b]^{b(g)} g$.
\item For $g\in H$, $\b g\bb=[\a,\b]^{-a(g)} g$.
\end{enumerate}
\end{lemma}

\begin{proof}
As $[\a,\b]\b=\a\b\ba\bb\b=\a\b\ba$, the first identity follows. As
$[\a,\b]$ is central, the second follows by taking a power.

Further, as $[\a,\b]^{-1}=\b\a\bb\ba$, we have
$\b\a\bb=[\a,\b]^{-1}\a$. Once more we take a power to get the next
identity.

Next, if $g\in H$, by Equation~\ref{eqrep},
$g=[\a,\b]^{\nu(g)}\a^{a(g)}\b^{b(g)}$, hence we can express its
conjugate by $\a$ as $\a g\ba= (\a [\a,\b]^{\nu(g)}\ba)( \a
\a^{a(g)}\ba)( \a \b^{b(g)}\ba)$. As $\a$ commutes with $[\a,\b]$ and
with powers of $\a$, using the previous identities and that $[\a,\b]$
is central, $\a g\ba=[\a,\b]^{\eta(g)} \a^{a(g)} [\a,\b]^{b(g)}\b^{b(g)}=
[\a,\b]^{b(g)}g$ as claimed. The proof of the remaining identity is
similar.
\end{proof}

We deduce the effect of \emph{canceling} a pair of letters in $g$
that are not adjacent on the Heisenberg invariant.

\begin{lemma}\label{canc0}
Let $g_1$ and $g_2$ be words in $H$.
\begin{enumerate}
\item $\nu(g_1\a g_2\ba)=b(g_2)+\nu(g_1g_2)$ 
\item $\nu(g_1\b g_2\bb)=-a(g_2)+\nu(g_1g_2)$
\end{enumerate}
\end{lemma}
\begin{proof}
By Lemma~\ref{ids}, as $[\a,\b]$ is central,
$g_1\a g_2\ba=g_1[\a,\b]^{a(g_2)}g_2=[\a,\b]^{b(g_2)}g_1g_2$. From this,
it follows that $\nu(g_1\a g_2\ba)=b(g_2)\nu(g_1g_2)$. The other identity
is similar.
\end{proof}

\section{Foldings and the Heisenberg invariant}

Consider henceforth a fixed word $g=l_1l_2\dots l_n$ of length $n$ in
$\a$, $\ba$, $\b$ and $\bb$. We shall regard the letters as cyclically
ordered, so that $l_{n+1}=l_1$. The word represents an RNA
strand. There is an obvious description of RNA secondary structures
without pseudo-knots in these terms.

\begin{definition}\label{deffold}
A \emph{folding} (or \emph{fold}) of the word $g$ is a collection of
\emph{disjoint} pairs $\F\subset\{(i,j):1\leq i,j\leq n,\ i\neq j\}$
such that
\begin{enumerate}
\item For $(i,j)\in \F$, either $l_i=\a$ and $l_j=\ba$ or $l_i=\b$ and 
$l_j=\bb$. 
\item For pairs $(i_1,j_1)\in \F$ and $(i_2,j_2)\in\F$,
$i_1<i_2<j_1\implies i_1<j_2<j_1$ and $i_1>i_2>j_1\implies
i_1>j_2>j_1$.
\end{enumerate}
\end{definition}

We denote the number of pairs in $\F$ by $\vert\F\vert$. The folding
process for RNA is governed by an attempt to maximise
$\vert\F\vert$. We define a pair $(i,j)\in\F$ to be an $\alpha$-pair
if $l_i=\alpha$ and a $\beta$-pair if $l_i=\b$. Every pair is either
an $\a$-pair or a $\b$-pair.

The condition on the pairs rules out \emph{nesting}(i.e., pseudo-knots). For a pair
$(i,j)\in \F$, we define the word $w(i,j)$ between the letters as
follows. If $i<j$, then $w(i,j)=l_{i+1}l_{i+2}\dots l_{j-1}$. If
$i>j$, then $w(i,j)=l_{i+1}l_{i+2}\dots l_nl_1l_2\dots l_{j-1}$. Thus,
this is the word from the $i$th letter to the $j$th letter in the
counterclockwise direction in the cyclic ordering. The word $D_{ij}(g)$
obtained from $g$ by \emph{canceling} the pair $(i,j)\in \F$ is the
word with $(n-2)$ letters obtained by deleting $l_i$ and $l_j$.

Using the conjugacy invariance of $\nu$, we can rephrase
Lemma~\ref{canc0} in the following way.

\begin{lemma}\label{cancel}
For $(i,j)\in \F$, if $l_i=\a$, then
$\nu(D_{ij}(g))=\nu(g)+b(w(i,j))$. If $l_i=\b$, then
$\nu(D_{ij}(g))=\nu(g)-a(w(i,j))$
\end{lemma}

Given a folding $\F=\{(i_1,j_1),(i_2,j_2),\dots,(i_k,j_k)\}$, we can
inductively compute $\nu(g)$ as follows. Let $g_0=g$ and let
$g_1=D_{i_1j_1}(g)$. Note that $\F$ induces a folding $\F_1$ of $g_1$
consisting of pairs of letters in $\F$ other than $(i_1,j_1)$. We can
thus continue inductively, deleting the pair corresponding to
$(i_2,j_2$). Thus, we get a sequence of words $g_0$,$g_1$, \dots,
$g_k$.

Let $c(i,j)=b(w(i,j))$ if $l_i=\a$ and $c(i,j)=-a(w(i,j) )$ if
$l_i=\b$. By inductively using Lemma~\ref{cancel}, we get the
following formula.

\begin{lemma}\label{fold}
We have
$$\nu(g)=\nu(g_k)+\sum_{l=1}^k c(i_l,j_l)$$
\end{lemma}
\begin{proof}
First note that by Lemma~\ref{cancel},
$\nu(g)=\nu(g_1)+c(i_1,j_1)$. We shall iterate this process. To do so,
observe that to apply Lemma~\ref{cancel} to $g_1$, we need to consider
the word $w_1(i_2,j_2)$ in $g_1$ between $i_2$ and $j_2$. By the no
nesting condition, this either equals $w(i_2,j_2)$ or differs from
this by a pair of letters which are inverses of each other. In either
case, $a(w_1(i,j))=a(w(i,j))$ and $b(w_1(i,j))=b(w(i,j))$. Thus, by
Lemma~\ref{cancel}, $\nu(g_1)=\nu(g_2)+c(i_2,j_2)$. We can now proceed
inductively in this fashion to prove the claim.
\end{proof}

An $\alpha$-folding is a folding such that for $(i,j)\in\F$,
$l_i=\alpha$. A $\beta$-folding is defined similarly. Consider a word
$g$ such that $a(g)=0$. A \emph{complete} $\alpha$-folding is an
$\alpha$-folding $\F$ such that if $l_i=\alpha$, then for some $j\neq
i$, $(i,j)\in\F$. The following proposition is straightforward.

\begin{proposition}\label{exafold}
If $a(g)=0$, there is a complete $\a$-fold for $g$.
\end{proposition}
\begin{proof}
If no letter is $\alpha$ (hence no letter is $\ba$ as $a(g)=0$), we
take $\F$ to be the empty set. Otherwise we can find a pair $(i,j)$
such that $w(i,j)$ does not contain the letters $\alpha$ and $\ba$. We
can proceed inductively by considering the word $g'=w(j,i)$. This
continues to satisfy the hypothesis and has fewer letters that are
$\alpha$ than in $g$. A complete $\a$-fold for $g'$ together with
$(i,j)$ gives a complete $\a$-fold for $g$.
\end{proof}

We can analogously define complete $\b$-folds. If $b(g)=0$ there are
complete $\b$-folds. Observe that for a complete $\a$-fold $\F$,
$\vert\F\vert$ is the number $n_{\a}$ of letters that are $\a$. Let
$n_{\b}$ similarly denote the number of letters that are $\b$. 

Lemma~\ref{fold} takes a particularly simple form for
complete $\a$-folds (and $\b$-folds).

\begin{lemma}\label{afld}
let $g$ be a word with $a(g)=b(g)=0$ and let $\F$ be a complete $\alpha$-fold
for $g$. Then using the notation of Lemma~\ref{fold}
$$\nu(g)=\sum_{l=1}^k b(w(i_l,j_l))$$
\end{lemma}
\begin{proof}
As $\F$ is a complete $\a$-folding, the word $g_k$ consists of the
letters of $g$ that are $\b$ or $\bb$. Thus, $g_k=\b^{b(g)}=1$ as
$b(g)=0$, hence $\nu(g_k)=0$. As $l_i=\a$ for all pairs $(i,j)\in \F$,
$c(i,j)=b(w(i,j))$. Thus the claim follows from Lemma~\ref{fold}
\end{proof}

\section{Computing the Heisenberg invariant}

The results of the previous section give an elementary formula for the
Heisenberg invariant from which it can be readily computed. We
formulate this below.

Let $g$ be a word in the free group $F$. Let $w(0,i)$ denote the
sub-word consisting of the $i-1$ letters preceding $l_i$. 

\begin{theorem}\label{comp}
The Heisenberg invariant is given by
$$\sum_{l_i=\ba} b(w(0,i))-\sum_{l_j=\a} b(w(0,j))$$
\end{theorem}
\begin{proof}
The proof consists of inductively applying Lemma~\ref{cancel} and
Proposition~\ref{unbal}. Consider the letters $l_{i_1}$,
$l_{i_2}$,\dots $l_{i_k}$ of $g$ that are either $\a$ or $\ba$. Let
$g_0=g$ and let $g_1$, $g_2$, \dots, $g_k$ be obtained from $g$ by
successively deleting the letters $l_{i_1}$, $l_{i_2}$, \dots
$l_{i_k}$. As $g_k$ is a power of $\beta$, $\nu(g_k)=0$. Thus, it
suffices to express $\nu(g_j)$ in terms of $\nu(g_{j+1})$. 

We first express $\nu(g_0)$ in terms of $\nu(g_1)$. Suppose
$l_{i_1}=\ba$. Then using Lemma~\ref{cancel} and Proposition~\ref{unbal},
we get
$$\nu(g_0)=\nu(\a g_0)=\nu(g_1)+b(w(0,i_1))$$
as $g_1$ is the result of canceling the first letter of $\a g_0$ with
the letter corresponding to $l_{i_1}$. Similarly, if $l_{i_1}=\a$, we
get 
$$\nu(g_0)=\nu(\ba g_0)=\nu(g_1)-b(w(0,i_1))$$
as can be readily deduced from Lemma~\ref{cancel}.

We now proceed inductively. We can use the same procedure as above to
relate $\nu(g_j)$ with $\nu(g_{j+1})$. As only letters that are $\a$ or
$\ba$ are deleted, the numbers $b(w(0,i_j))$ are not altered during
the inductive construction. Hence we still have 
$$\nu(g_j)=\nu(g_{j+1})\pm b(w(0,i_j))$$
with the sign determined by whether $l_{i_j}$ is $\a$ or $\ba$. Using
the formula recursively gives the claim.
\end{proof}

\section{A rigidity theorem}

Consider a word $g$ with $a(g)=b(g)=0$. Let $n_\alpha$ be the number
of letters in $g$ that are $\a$ (and hence the number of letters that
are $\ba$) and let $n_\b$ be the number of letters that are $\b$. Then
the number of letters $n$ of $g$ is $2(n_\a+n_\b)$.

\begin{theorem}\label{rigid}
We have $\nu(g)\leq \left(\frac{n}{4}\right)^2$ with equality if and
only if $n_\a=n_\b=\frac{n}{4}$ and $g$ is conjugate to the commutator
$[\a^{n/4},\b^{n/4}]$.
\end{theorem}
\begin{proof}
By Lemma~\ref{exafold}, there is a complete $\a$-fold $\F$ for $g$. By
Lemma~\ref{afld}, we have
$$\nu(g)=\sum_{l=1}^k b(w(i_l,j_l))$$ Observe that as $w(i,j)$ is a
sub-word of $g$, $b(w(i.j))\leq n_\b$. The number of terms in the above
sum is $n_\a$. Hence it follows that $\nu(g)\leq n_\a n_\b$. As the
geometric mean is at most the arithmetic mean and $n=2(n_\a+n_\b)$, it
follows that
$$\nu(g)\leq n_\a n_\b\leq \left(\frac{n}{4}\right)^2$$

In case of equality, each of the above inequalities must be an
equality. Hence as the arithmetic mean equals the geometric mean,
$n_\a=n_\b=n/4$. Further, for each pair $(i,j)$, $b(w(i,j))=n_\b$. it
follows that the letters in $w(i,j)$ include all the letters in $g$
that are $\b$ and none of the letters that are $\bb$. It is easy to
deduce that $g$ is conjugate to $[\a^{n/4},\b^{n/4}]$.
\end{proof}

\section{The Heisenberg invariant and Local minima}

To motivate our next (and most interesting) result, consider the word
$[\a^{n/4},\b^{n/4}]$ (or more generally a word of the form
$[\alpha^k,\beta^l]$). It is easy to see that any folding of this word
is either an $\a$-folding or a $\b$-folding. Hence to pass from a
complete $\a$-folding to a complete $\b$-folding at some intermediate
stage, the RNA strand must be completely unfolded.

The word $[\a^{n/4},\b^{n/4}]$ is characterised by the Heisenberg
invariant. However, this situation is too special. We show that if the
Heisenberg invariant is close to the maximum value in an appropriate
sense, then we have foldings $\F_1$ and $\F_2$ so that while passing
from $\F_1$ to $\F_2$, at some intermediate stage there are
significantly fewer pairs than in $\F_1$ and $\F_2$. To make this
precise we introduce some notation.

Fix a word $g$ with $a(g)=b(g)=0$ and let $n$, $n_\a$ and $n_\b$ be as
before. For a folding $\F$ of $g$, we define the deficiency $\rho(\F)$
as
$$\rho(\F)=1-\frac{2\vert\F\vert}{n}$$

This is the fraction of letters that are not in some pair. The
potential energy can be assumed to be a monotonically increasing
function of the deficiency.

We say that foldings $\F$ and $\F'$ are \emph{adjacent} if their
symmetric difference consists of a single pair. A \emph{path} from
$\F$ to $\F'$ is a sequence of foldings $\F=\F_0$, $\F_1$,
$\F_2$,\dots $\F_k=\F'$ such that $\F_l$ is adjacent to $\F_{l+1}$ for
$0\leq l<k$. This represents a sequence of steps by which an RNA
molecule can pass between one folding and another.

We can now state our result. Assume that $k$ is an integer with $k/n$
small. Let $\F$ be a complete $\alpha$-pairing for $g$ and $\F'$ a
complete $\b$-pairing. We show that if $\nu(g)$ is close to maximal,
then for any path from $\F$ to $\F'$, some path has much larger
deficiency than both $\F$ and $\F'$.

\begin{theorem}\label{saddle}
Suppose $\nu(g)\geq\left(\frac{n}{4}\right)^2-k^2$. Then,
\begin{enumerate}
\item $\rho(\F)\leq \frac{1}{2}+\frac{2k}{n}$ and $\rho(\F')\leq
\frac{1}{2}+\frac{2k}{n}$
\item Given any path $\F_0=\F$, $\F_1$, $\F_2$,\dots $\F_k=\F'$ from
  $\F$ to $\F'$, for some $l$, $1\leq l<k$,
  $\rho^2(\F_l)\geq\frac{1}{2}-24\left(\frac{k}{n}\right)^2$.
\end{enumerate}
\end{theorem}
\begin{remark}
If we assume that $k/n$ is small, then Theorem~\ref{saddle} says that
deficiency of the foldings $\F$ and $\F'$ are not much more than
$1/2$, while that of some intermediate folding is not much less than
$1/\sqrt{2}$. Thus, the deficiency, and hence the potential energy,
increases significantly in passing from $\F$ to $\F'$
\end{remark}

\begin{proof}[Proof of Theorem~\ref{saddle}]
Assume henceforth that the hypothesis of Theorem~\ref{saddle} is
satisfied. Our first step is as in Theorem~\ref{rigid}.

\begin{lemma}\label{balance}
We have $n_\a\geq n/4-k$ and $n_\b\geq n/4-k$.
\end{lemma}
\begin{proof}
Let $n_\a=n/4-p$ with $p$ an integer. Then $n_\b=n/4+p$. As
$a(g)=b(g)=0$, applying Theorem~\ref{afld} for the folding $\F$, we get 
$$\nu(g)=\sum_{l=1}^k b(w(i_l,j_l))$$ As $\F$ is a complete $\a$-fold,
the number of terms of the above sum is $n/4-p$. As each word $w(i,j)$
is a sub-word of $g$, $b(w(i,j))\leq n_\b=n/4+p$. Thus,
$$\nu(g)\leq \left(\frac{n}{4}-p\right)\left(\frac{n}{4}+p\right) =
\left(\frac{n}{4}\right)^2-p^2$$

As $\nu(g)\geq\left(\frac{n}{4}\right)^2-k^2$, $p^2\leq k^2$, from
which the lemma follows.
\end{proof}

As $\F$ is a complete $\a$-pairing, $\vert\F\vert=n_\a$. Hence, as
$n_\a\geq n/4-k$, an easy calculation shows that $\rho(\F)\leq
\frac{1}{2}+\frac{2k}{n}$ as claimed. Similarly, $\rho(\F')\leq
\frac{1}{2}+\frac{2k}{n}$.

We now turn to the second part of the Theorem. We shall first make
some observations. Let $\F''$ be a folding of $g$ with $m_\a$ pairs
involving $\a$ and $m_\b$ pairs involving $\b$. We shall use the
formula of Lemma~\ref{fold}, namely
$$\nu(g)=\nu(g_k)+\sum_{l=1}^k c(i_l,j_l)$$

Consider first the sum $\sum_{l=1}^k c(i_l,j_l)$. This has two kinds
of terms corresponding to $\a$-pairs and $\b$-pairs. We consider these
separately. First consider the $m_\a$-terms corresponding to
$(i,j)\in\F''$ with $l_i=\a$. Then $c(i,j)=b(w(i,j))$.

The key observation is that $b(w(i,j))\leq n_\b-m_\b$. This is because
$b(w(i,j))$ is the difference between the number of letters in $w(i,j)$
that are $\b$ and the number that are $\bb$. As there is no nesting,
if $(i',j')\in \F''$ with $l_{i'}=\b$ contained in $w(i,j)$, then
$l_{j'}=\bb$ is also contained in $w(i,j)$. Hence the net contribution
of the letters of the pair to $b(w(i,j))$ is zero. 

Thus, the sum of the terms corresponding to pairs $(i,j)$ with
$l_i=\a$ is at most $m_\a(n_\b-m_\b)$. We have a similar result for the
$\beta$-pairs.

Now consider a path as in the hypothesis of the Theorem. It is easy to
see that for some $l$, the number of $\a$-pairs in $\F_l$ is equal to
the number of $\b$-pairs. Denote this number by $m$. We shall find a
lower bound for $\rho(\F_l)$. Let the number of elements that are not
in any pair be $d$. If $g_k$ is as above, then $g_k$ has $d$ elements.

Consider the equation 
\begin{equation}\label{sum}
\nu(g)=\nu(g_k)+\sum_{l=1}^k c(i_l,j_l)
\end{equation}
for $\F_l$. As above, we consider separately the terms corresponding
to $\a$-pairs and $\b$-pairs. Let $p$ be as in the proof of
Lemma~\ref{balance}. We have seen that the total contribution of the
$\a$-pairs is at most $m(n_\b-m)=m(n/4+p-m)$. By the inequality
between arithmetic and geometric means,
$$m(n/4+p-m)\leq\frac{1}{4}\left(\frac{n}{4}+p\right)^2$$
Similarly, the contribution of the $\b$-terms can be bounded by 
$$m(n/4-p-m)\leq\frac{1}{4}\left(\frac{n}{4}-p\right)^2$$
Hence we get an upper bound on the sum 
\begin{equation}\label{upbound}
\sum_{l=1}^k c(i_l,j_l)\leq
\frac{1}{2}\left(\left(\frac{n}{4}\right)^2+k^2\right)
\end{equation}
using $p^2\leq k^2$ as
in the proof of Lemma~\ref{balance}.

By hypothesis, $\nu(g)\geq\left(\frac{n}{4}\right)^2-k^2$. Hence, from
Equation~\ref{sum} and Equation~\ref{upbound} we get
\begin{equation}\label{lowbound}
\nu(g_k)\geq \frac{1}{2}\left(\left(\frac{n}{4}\right)^2-3k^2\right)
\end{equation}

By Theorem~\ref{rigid}, $\nu(g_k)\leq(d/4)^2$, where $d$ is the number of
unpaired letters. By the definition of deficiency,
$d=n\rho(\F_l)$. Thus, from Equation~\ref{lowbound} we obtain the
inequality
$$\left(\frac{n\rho(\F_l)}{4}\right)^2\geq
\frac{1}{2}\left(\left(\frac{n}{4}\right)^2-3k^2\right) $$ or
$$\rho^2(\F_l)\geq\frac{1}{2}-24\left(\frac{k}{n}\right)^2$$
as claimed.

\end{proof}

One does not expect in general for the condition $a(g)=b(g)=0$ to be
satisfied. However, one can use Theorem~\ref{saddle} provided $a(g)$
and $b(g)$ are small compared to $n$ and
$\nu(g)\geq\left(\frac{n}{4}\right)^2-k^2$ with $k/n$ small. To do
this, consider the word $g'=\a^{-a(g)}g\b^{-b(g)}$. By
Proposition~\ref{unbal}, $\nu(g')=\nu(g)$.

The word $g'$ has length $n'=n+a(g)+b(g)$. The hypothesis
$\nu(g)\geq\left(\frac{n}{4}\right)^2-k^2$ can be rephrased as 
$$\nu(g')\geq\left(\frac{n'}{4}\right)^2-{k'}^2$$
with
$${k'}^2=k^2+\frac{n'(a(g)+b(g))}{8}-\left(\frac{a(g)+b(g)}{4}\right)^2$$

If $k$, $a(g)$ and $b(g)$ are all small compared to $n$, then it
follows that $k'$ is small compared to $n'$ (as well as compared to
$n$). Hence we can apply Theorem~\ref{saddle} in this case (to the
element $g'$ and deduce for $g$).  

We can see that generically $a(g)$ and $b(g)$ are comparable to
$\sqrt{n}$. Namely, $a$ and $b$ for a random string can be veiwed as
the results of (independent) one-dimensional random walks with
$n_\alpha$ and $n_\beta$ steps. Hence $a(g)$ and $b(g)$ are
generically of the order of $\sqrt{n_\alpha}$ and $\sqrt{n_\beta}$,
respectively.

Results analogous to those of this section continue to hold if we
modify our model so that nearby bases do not pair (as is the case
biologically). As further restrictions only increase the number of
unpaired bases (in particular of intermediate stages), it suffices to
give lower bounds for the number of paired bases (hence upper bounds
for the number of unpaired bases) for appropriate $\alpha$-foldings
and $\beta$-foldings. An examination of our proof yields such a
bound. Namely, if the Heisenberg invariant is close to its maximal
value $(n/4)^2$, then (in a sense that can be made precise) for a
complete $\alpha$-folding most terms $b(w(i_l,j_l)$ in
Lemma~\ref{afld} must be close to $n/4$. This means that most bonds in
the complete $\alpha$-folding are between bases that are not close to
each other. Hence we can obtain a lower bound on the number of bonds
in an $\alpha$-folding without nearby bases paired. The case of
$\beta$-foldings is similar.

\section{Lower bounds on deficiency}

We now turn to another important application of the Heisenberg
invariant. Let $g$ be a word in the free group $F$ with
$a(g)=b(g)=0$. As the Heisenberg invariant is a measure of
non-triviality, one expects that if $\nu(g)$ is large then the
deficiency of any folding $\F$ of $g$ is large. We now prove such a
result. 

\begin{theorem}
For any folding $\F$ of $g$, if $\rho=\rho(\F)$ 
$$\rho-\frac{3\rho^2}{4}\geq \frac{4\nu(g)}{n^2}$$
\end{theorem}
\begin{proof}
As before, we use the formula of Lemma~\ref{fold}, namely
$$\nu(g)=\nu(g_k)+\sum_{l=1}^k c(i_l,j_l)$$

Let $d=n\rho(g)$ be the number of letters that are unpaired in $\F$
(hence the number of letters of $g_k$). We claim that $c(i_l,j_l)\leq
d/2$. This follows if the number of unpaired letters of the sub-word
$w(i_l,j_l)$ is at most $d/2$ as pairs have canceling contributions
to $a(w(i_l,j_l))$ and $b(w(i_l,j_l))$ (as in the Proof of
Theorem~\ref{saddle}). Otherwise $w(j_l,i_l)$ has at most $d/2$
unpaired letters, from which we can deduce the result as
$a(w(j_l,i_l))=-a(w(i_l,j_l))$ and $b(w(j_l,i_l))=-b(w(i,j))$ since
$a(g)=b(g)=0$ and $l_{i_l}$ and $l_{j_l}$ are a canceling pair.

The number of pairs in $\F$ is $(n-d)/2$. Further, by
Theorem~\ref{rigid}, $\nu(g_k)\leq (\frac{d}{4})^2$. Hence, as
$d=n\rho(g)$, we get
$$\nu(g)\leq \left(\frac{d}{4}\right)^2+\frac{d(n-d)}{4}=
\frac{n^2}{4}\left(\rho(\F)-\frac{3\rho^2}{4}\right)$$
from which the claim follows.
\end{proof}

\section{The Milnor invariants and other extensions}\label{S:mil}

We now turn to the higher Milnor invariants. We shall be very sketchy
in this section as our goal is to indicate further extensions of our
methods.

We begin by recalling the conceptual scheme for defining the Milnor
invariants. Suppose $g$ is an element of the free group $F$. The
Milnor invariants (in our situation) are measures of how different $g$
is from the identity. Further, these are invariant under conjugacy and
are additive. Thus, they are very well suited for studying RNA
folding.

The Milnor invariants are defined in terms of the lower central
series. The first Milnor invariants correspond to the image of $g$ in
the abelianisation $F/F_2$. If this is trivial, then $g$ is in $F_2$,
and we consider its image in $F_2/F_3$. This is the Heisenberg
invariant. 

We proceed inductively, with the Milnor invariants of order $n$
defined if those of lower order vanish. Namely, if the invariants of
order less than $n$ vanish, then $g\in F_n$. We consider the image of
$g$ in $F_n/F_{n+1}$, which is a finitely generated free abelian group. By
choosing a basis for this group, we get finitely many integers
determining the image. Thus, we have an infinite series of
invariants. It follows from the Magnus expansion homomorphism (which
we recall below) that $\bigcap_{i=1}^\infty F_i=\{1\}$, hence all the
Milnor invariants of $g$ vanish if and only if $g$ is trivial.

In our situation, we would like to define the Milnor invariants of
order $n$ up to an error corresponding to the lower order Milnor
invariants, as we did in the case of the Heisenberg invariant. We can
do this as we have a given basis for the free group $F$. Namely, we
choose and fix sections $s_k:F/F_k\to F$, i.e. functions between the
underlying sets of $F/F_k$ and $F$ (not preserving the algebraic
structure) so that for the quotient homomorphism $q_k:F\to F/F_k$, we
have $q_k\circ s_k:F/F_k\to F/F_k$ is the identity. Then for an
arbitrary $g\in F$, $g\cdot(s_n\circ q_n(g))^{-1}$ is in $F_n$, so we
can define Milnor invariants by considering the image of $g(s_n\circ
q_n(g))^{-1}$ in $F_n/F_{n+1}$. This is what we did in the case of the
Heisenberg invariant here, using the section $(a,b)\mapsto
\a^a\b^b$. As we are not studying the higher invariants in detail
here, we do not construct explicit sections.

We remark that in the original topological context, there is no
canonical basis and hence one cannot make a well-defined choice of
section. However, the Milnor invariants are defined modulo those of
lower order. Such arithmetical considerations are not likely to be
fruitful in the context of RNA. 

An explicit description of the Milnor invariants, allowing for
efficient computation, can be given in terms of the so called
\emph{Magnus expansion homomorphism}. Consider formal power series
with integer coefficients in two non-commuting variables $X$ and $Y$
(i.e., where $XY$ is not $YX$ for example). These form a ring
$\Z[[X,Y]]$. The Magnus expansion homomorphism $M$ is an injective
homomorphism from the free group $F$ to the multiplicative group of
invertible elements of $\Z[[X,Y]]$. As is well known, the invertible
elements of $\Z[[X,Y]]$ are the formal power series with constant term
$1$. The Magnus homomorphism is the unique homomorphism such that the
image of the generators $\a$ and $\b$ satisfy
$$M(\a)=1+X$$
$$M(\b)=1+Y$$

By properties of formal power series, we have
$$M(\ba)=1-X+X^2-X^3+\dots$$
$$M(\bb)=1-Y+Y^2-Y^3+\dots$$

In general, for a word $g$, we multiply the images of the letters of
$g$ to obtain $M(g)$. 

For $g\in F$, the constant term of $M(g)$ is $1$. The coefficients of
$X$ and $Y$ are $a(g)$ and $b(g)$, and thus determine and are
determined by the abelianisation. In particular, the coefficients
vanish if and only if $a(g)=b(g)=0$.

Suppose $a(g)=b(g)=0$. Then it is known that the term of degree two is
of the form $c(XY-YX)$. The Heisenberg invariant is $\nu(g)=c$. More
generally, $g\in F_n$ if and only if the terms of degree less than $n$
(except the constant term) vanish. As a consequence, we have a
well-defined homomorphism of abelian groups from $F_n/F_{n+1}$ to
homogeneous polynomials of degree $n$ in the non-commuting variables
$X$ and $Y$. This is injective but not surjective. The image consists
of so called \emph{Lie elements}, for which an explicit basis can be
given.

Thus, Milnor invariants can be defined as linear combinations of
coefficients of the polynomial $M(g)$. This expression is not unique
as the Lie elements do not generate all the homogeneous polynomials of
degree $n$. For instance, the Heisenberg invariant can be defined as
the coefficient of $XY$ or the negative of the coefficient of $YX$. If
$g$ is in $F_n$, then all these expressions give the same value. For a
general element $g$, the choice of a linear combination is analogous
to the choice of a section $s_n$ in the earlier description. As in the
case of the sections $s_n$, we do not give details of the relevant
bases and coefficients.

We finally turn to the original topological context which motivated
this work. A \emph{link} $L$ in a smooth $3$-dimensional manifold is a
disjoint union of smoothly embedded circles. An \emph{unlink} is a
link $L$ so that the components of $L$ bound disjoint, smoothly
embedded discs. There are three well studied relations on links -
those of isotopy, homotopy and concordance. We consider links up to
these relations.

Firstly, let $L'$ be an unlink with two components. The fundamental
group of the complement is a free group on two generators, and hence
can be identified with $F$. Given a word $g$ in $F$, consider a
three-component link $L$, with the third component a curve $\gamma$
representing the element $g$ (up to conjugation) in the complement of
$L'$. This is not well-defined, but we may study properties of this
link that depend only on $g$.

The first observation is that the link is defined up to \emph{link
homotopy}, i.e., changing the link through a family so that each
component is allowed to cross itself but not others. For such
three-component links, the link is determined by exactly three
invariants. These correspond to $a(g)$, $b(g)$ and $\nu(g)$. We get
further invariants by observing that we can keep the first two
components fixed and only allow $\gamma$ to cross itself. We expect
that it will be fruitful to study the link $L$ up to concordance
(which implies link homotopy). Link homotopy, isotopy and concordance
have been extensively studied. We refer to~\cite{CO}--\cite{St} for
some of the fundamental results concerning these.

\section{Concluding remarks}

In this article, we have constructed an easily computed number, the
Heisenberg invariant, associated to a strand of RNA, and showed that
it is related to a dynamically important property of Watson-Crick
pairing. Values of the invariant large enough to yield dynamically
interesting results are not generic. However, in biological systems it
is important to understand molecules with special properties, as
evolution ensures that living systems are not generic but have
compositions giving desirable properties.

We end with a conjecture regarding the Milnor invariants of order one
more than the Heisenberg invariant. We have seen that the Heisenberg
invariant being large implies that we have minima that are widely
separated. We expect that the next invariants being large implies that
there are minimax paths between minima that are widely separated. This is
significant in the context of catalysis.

\bibliographystyle{amsplain}

\begin{thebibliography}{10}

\bibitem{CO} Cochran, Tim D.; Orr, Kent E. \textit{Not all links are
concordant to boundary links.}  Ann. of Math. (2) 138 (1993), no. 3,
519--554.

\bibitem{COT} Cochran, Tim D.; Orr, Kent E.; Teichner, Peter \textit{Knot
concordance, Whitney towers and $L\sp 2$-signatures.}  Ann. of
Math. (2) 157 (2003), no. 2, 433--519.

\bibitem{Ca} Casson, A. J.\textit{Link cobordism and Milnor's
  invariant.}  Bull. London Math. Soc.  7 (1975), 39--40.

\bibitem{FT} Freedman, Michael H.; Teichner, Peter
\textit{$4$-manifold topology. I. Subexponential groups.}
Invent. Math.  122 (1995), no. 3, 509--529.

\bibitem{Gi} Giffen, Charles H.
\textit{Link concordance implies link homotopy.}
Math. Scand. 45 (1979), no. 2, 243--254.

\bibitem{Go} Goldsmith, Deborah L. \textit{Concordance implies
homotopy for classical links in $M\sp{3}$.}  Comment. Math. Helv.  54
(1979), no. 3, 347--355.

\bibitem{Kr} Krushkal, Vyacheslav S. \textit{Additivity properties of Milnor's
$\overlineµ$-invariants.}  J. Knot Theory Ramifications 7 (1998),
no. 5, 625--637.

\bibitem{Le} Levine, J. P. Surgery on links and the $\overlineµ$-invariants.
Topology 26 (1987), no. 1, 45--61.


\bibitem{Mi1} Milnor, John \textit{Link groups.}  Ann. of Math. (2)
59, (1954). 177--195

\bibitem{Mi2} Milnor, John \textit{Isotopy of links.} Algebraic
geometry and topology.  A symposium in honor of S. Lefschetz,
pp. 280--306. Princeton University Press, Princeton, N. J., 1957.

\bibitem{St} Stallings, John \textit{Homology and central series of
groups.}  J. Algebra 2 (1965) 170--181.


\bibitem{RNA} RF Gesteland, JF Atkins \textit{The RNA world: the
nature of modern RNA suggests a prebiotic RNA world} 1993 - Cold
Spring Harbor Laboratory Press

\bibitem{FW} Fox,G.E. and Woese,C.R. \textit{ RNA secondary
structure.} Nature, 256, p. 505.

\bibitem{GG} Gardner P P, Giegerich R A, \textit{comprehensive
comparison of comparative RNA structure prediction approaches}, BMC
BIOINFORMATICS 5: Art. No. 140 SEP 30 2004


\bibitem{c1} Ding Y, \textit{A statistical sampling algorithm for RNA
secondary structure prediction,} NUCLEIC ACIDS RESEARCH 31 : 7280 2003

\bibitem{c2} Dirks R M, \textit{A partition function algorithm for
nucleic acid secondary structure including pseudoknots} JOURNAL OF
COMPUTATIONAL CHEMISTRY 24 : 1664 2003

\bibitem{c3} Fields D S, \textit{An analysis of large rRNA sequences
folded by a thermodynamic method} FOLDING \& DESIGN 1 : 419 1996

\bibitem{c4} Gorodkin J, \textit{Discovering common stem-loop motifs in
unaligned RNA sequences} NUCLEIC ACIDS RESEARCH 29 : 2135 2001

\bibitem{c5} Higgs P G, \textit{RNA secondary structure: physical and
computational aspects} QUARTERLY REVIEWS OF BIOPHYSICS 33 : 199 2000

\bibitem{c6} Hofacker I L, \textit{Alignment of RNA base pairing probability
matrices} BIOINFORMATICS 20 : 2222 2004

\bibitem{c7} Jeffares DC, \textit{Relics from the RNA world}, JOURNAL OF
MOLECULAR EVOLUTION 46 : 18 1998

\bibitem{c8} Ji Y M, \textit{A graph theoretical approach for predicting
common RNA secondary structure motifs including pseudoknots in
unaligned sequences} BIOINFORMATICS 20 : 1591 2004

\bibitem{c9} Knudsen B, \textit{RNA secondary structure prediction using
stochastic context-free grammars and evolutionary history}
BIOINFORMATICS 15 : 446 1999

\bibitem{c10} Moulton V \textit{Metrics on RNA secondary structures}
JOURNAL OF COMPUTATIONAL BIOLOGY 7 : 277 2000

\bibitem{c11} Nussinov R, \textit{Algorithms for loop matchings} SIAM JOURNAL
ON APPLIED MATHEMATICS 35 : 68 1978

\bibitem{c12} Rivas E, \textit{The language of RNA: a formal grammar that
includes pseudoknots} Bioinformatics 16 : 334 2000

\bibitem{Ti} Tinoco I, \textit{How RNA folds} Journal of molecular biology 293
: 271 1999

\bibitem{c14} Wang Z Z, \textit{Alignment between two RNA structures}
Mathematical foundations of computer science 2001 2136 : 690 2001





\end{thebibliography}

\end{document}